\documentclass[article,11pt]{amsart}
\usepackage[margin = 1in]{geometry}
\usepackage{latexsym}
\usepackage{graphicx}
\usepackage{mathptmx}
\usepackage{textcomp}
\usepackage{amsmath}
\usepackage{amsfonts}
\usepackage{amssymb}
\usepackage{amsbsy}
\usepackage{amsthm}
\usepackage{color, float, subcaption, xcolor, epstopdf}
\usepackage{chicago}
\usepackage{hyperref}

\newcommand{\oo}{\infty}
\newcommand{\R}{\mathbb{R}}
\newcommand{\T}{\intercal}
\DeclareMathOperator{\argmin}{argmin}

\newtheorem{theorem}{Theorem}
\newtheorem{corollary}[theorem]{Corollary}
\newtheorem{proposition}[theorem]{Proposition}
\newtheorem{lemma}[theorem]{Lemma}

\begin{document}

\title{Portfolio Optimization with Expectile and Omega Functions}

\author{Alexander Wagner}
\address{Department of Mathematics, University of Florida}
\email{wagnera@ufl.edu}
\urladdr{https://people.clas.ufl.edu/wagnera/}

\author{Stan Uryasev}
\address{Department of Industrial and Systems Engineering, University of Florida}
\email{uryasev@ise.ufl.edu}
\urladdr{http://www.ise.ufl.edu/uryasev/}

\begin{abstract}

This paper proves equivalences of portfolio optimization problems with negative expectile and omega ratio. We derive subgradients for the negative expectile as a function of the portfolio from a known dual representation of expectile and general theory about subgradients of risk measures. We also give an elementary derivation of the gradient of negative expectile under some assumptions and provide an example where negative expectile is demonstrably not differentiable. We conducted a case study and solved portfolio optimization problems with negative expectile objective and constraint.
\end{abstract}

\noindent \textcopyright 2019 IEEE.  Personal use of this material is permitted.  Permission from IEEE must be obtained for all other uses, in any current or future media, including reprinting/republishing this material for advertising or promotional purposes, creating new collective works, for resale or redistribution to servers or lists, or reuse of any copyrighted component of this work in other works.\\

\maketitle

\section{Introduction}
\label{sec:intro}

In this paper, portfolio optimization is the task of maximizing the expected return of a portfolio subject to the risk of the portfolio not exceeding a prespecified level. Formally, let $\xi = (\xi_1,\dots, \xi_n)$ be a vector containing random variables representing the returns of $n$ financial instruments. Let $x = (x_1, \dots, x_n ) \in \R^n$. The random return of the portfolio is given by $\xi^\T x$. The decision variable $x_i$ represents position being invested in the $i$-th instrument. Let us denote the risk measure of interest by $R$. We investigate two single-stage stochastic optimization problems, 
\begin{align*}
	&\max E[\xi^\T x] \quad \text{subject to} \quad R(\xi^\T x) \leq b,
	\ x \in V\in \R^n \; \text{ and} \\
	&\min R(\xi^\T x) \quad \text{subject to} \quad E[\xi^\T x] \geq r,
	\ x \in V\in \R^n\;,
\end{align*}
where $V\in \R^n$ is a convex set of feasible portfolios.

This paper considers two closely related risk measures, negative expectile and omega.
Expectiles are asymmetric generalizations of expected value introduced by \citeN{newey1987asymmetric} and defined by \eqref{expectile}. When $q \leq \frac{1}{2}$, negative level-$q$ expectiles satisfy the conditions of a coherent risk measure in the sense of \shortciteN{artzner1999coherent}. Section \ref{sec:subgrad} of this paper shows that for  $q < \frac{1}{2}$ the  level-$q$ negative expectile is strictly expectation bounded (i.e. averse) as defined by \shortciteN{rockafellar2006generalized}.
 Moreover, negative expectile is elicitable, that is there exists a consistent way of comparing the quality of procedures that predict its value \cite{ziegel2016coherence}. Indeed, \shortciteN{Bellini2019} have suggested a procedure for testing the accuracy of an expectile forecasting model. A conditional version of expectile naturally generalizing conditional expectation has also been developed \shortcite{BELLINI2018117}.

In addition to this list of desirable theoretical properties, \citeN{bellini2017risk} offer the following intuitive motivation for the use of negative expectile as a measure of risk for portfolio optimization. Value at risk is a commonly used risk measure and can be given by \eqref{varacceptance} employing the concept of acceptance sets. The same reformulation can be given to negative expectile as in \eqref{nexpectileacceptance}.

\begin{align}
\text{VaR}_{\alpha}(X) = \inf \{m\in\R \ |\ P(X+m > 0) \geq \frac{1-\alpha}{\alpha}P(X+m\leq 0)\} \label{varacceptance} \\
R_q(X) = \inf \{m \in \R \ |\ E[(X+m)_+]\geq \frac{1-q}{q}E[(X+m)_-]\} \label{nexpectileacceptance}
\end{align}

So if VaR represents the smallest amount we need to add to the random return such that the \textit{probability} of a gain is $(1-\alpha)/\alpha$ times greater than the \textit{probability} of a loss, then negative expectile represents the smallest amount we need to add such that the \textit{expected} gains are $(1-q)/q$ times greater than the \textit{expected} losses. 

Property \eqref{nexpectileacceptance} of the negative expectile  has a clear relationship to the ratio of expected over-performance to expected under-performance relative to a constant benchmark $B$. This quantity is denoted $\Omega_B$ and was introduced by \citeN{keating2002universal}, 
\[
\Omega_B(X) = \frac{E[(X-B)_+]}{E[(X-B)_-]}.
\]

Choosing portfolio weights to maximize omega is a non-convex problem \shortcite{kane2009optimizing}. \shortciteN{mausser2006optimising} reduce the omega maximization problem to a linear programming problem in the case that the optimal portfolio has an expected return exceeding the benchmark. The case study ``Omega Portfolio Rebalancing" \cite{casestudyomega}
reduces omega optimization to maximization of expected return with a convex constraint on a partial moment. This case study posted data for a test problem  (provided by a mutual fund)  and codes in Text, MATLAB, and R formats. 
In other work, \shortciteN{kapsos2014optimizing} show that the omega maximization problem can be reformulated equivalently as a quasi-convex optimization problem. Section~\ref{sec:omega} shows equivalence of expectile and omega portfolio optimization problems.  The relation of expectile to omega has also been explored by \shortciteN{Bellini2018} where they define a stochastic order based on expectile and show its equivalence to the same inequality holding for omega at every benchmark.

\citeN{jakobsons2016} provides three linear programming formulations for portfolio optimization with an expectile objective when asset returns have a finite, discrete distribution. In the case of a large number of scenarios, convex programming has a significant advantage compared to linear programming. While the set of constraints for linear programming may grow infeasibly large, the convex formulation need only compute a subgradient. To this end, we deduce a formula for subgradients of negative expectile in Section \ref{sec:subgrad} by combining a known dual representation of expectile \shortcite{bellini2014generalized} and general theory about subgradients of risk measures \shortcite{rockafellar2006optimality}. Section \ref{sec:applications} provides numerical results with Portfolio Safeguard subroutines calculating negative expectile and subgradient in order to solve portfolio optimization problems with a convex programming approach. Section \ref{sec:derivative} provides an elementary derivation of the gradient of negative expectile as a function of the portfolio and ends with an example where the negative expectile is demonstrably not differentiable.

\section{Background}

This paper considers a general probability space $(\mathcal A,\mathcal F, P)$, unless specified otherwise. In some application-focused derivations, it is noted that the space is assumed to be finite. Inequalities between random variables are meant almost surely.
Expectiles are generalizations of expected value introduced by \citeN{newey1987asymmetric}. 
Let $X$ belong to the space $\mathcal L^2(\mathcal A)$ of random variables with finite second moment.
The level-$q$ expectile of $X$ is denoted $e_q(X)$ and defined according to \eqref{expectile}, where $x_+ = \max(x,0)$ and $x_- = \max(-x,0)$.
 
 \begin{equation}
 \label{expectile}
 e_q(X) = \argmin_{m\in\R}\ qE[(X-m)_+^2] + (1-q)E[(X-m)_-^2],\ q\in(0,1)\;.
 \end{equation} 
 \citeN{newey1987asymmetric} showed that $m = e_q(X)$ is the unique solution of the following equation \eqref{focexpectile},
 
 \begin{equation}
 \label{focexpectile}
 qE[(X-m)_+] - (1-q)E[(X-m)_-] = 0\;.
 \end{equation} 
 We take this as an alternative definition of expectile that extends to $\mathcal L^1(\mathcal A)$, the space of random variables with finite first moment. For calculating expectiles, it is enough to assume the existence of the first moment. We make use of the following lemma throughout the paper.
 
\begin{lemma}
\label{monoexpectile}
	If $0 < q < 1$, then $qE[(X-m)_+] - (1-q)E[(X-m)_-]$ is strictly decreasing in $m$.
\end{lemma}
 \begin{proof}
 	If $m < n$, then $(X-m) > (X-n)$ which implies $E[(X-m)_+] \geq E[(X-n)_+]$ and $E[(X-m)_-] \leq E[(X-n)_-]$.
If both inequalities are equalities then
\[
E[X-m] = E[(X-m)_+] - E[(X-m)_-] = E[(X-n)_+] - E[(X-n)_-] = E[X-n]
\]
This is a contradiction, so either $E[(X-m)_+] > E[(X-n)_+]$ or $E[(X-m)_-] < E[(X-n)_-]$, from which the result follows since $q, (1-q) > 0$.
 \end{proof}

We define $R_q(X) = -e_q(X)$ for some fixed $q$. \shortciteN{bellini2014generalized} have shown that for $q \leq \frac{1}{2}$, $R_q(X)$ satisfies the following axioms of a coherent risk measure. 

\begin{enumerate}
	\item $R_q(X + C) = R_q(X) - C$ for all $X$ and constants $C$,
	\item $R_q(0) = 0$, and $R_q(\lambda X) = \lambda R_q(X)$ for all $X$ and $\lambda > 0$,
	\item $R_q(X + Y) \leq R_q(X) + R_q(Y)$ for all $X,Y$,
	\item $R_q(X) \leq R_q(Y)$ when $X \geq Y$ almost surely.
\end{enumerate}

 In this paper, we consider $X = \xi^\T x$, where  $\xi = (\xi_1,\dots, \xi_n)$ is a vector containing random variables representing the returns of $n$ financial instruments and the decision variable $x_i$ represents position being invested in the $i$-th instrument. We denote by $V\in \R^n$ a convex set of feasible portfolios. Section~\ref{sec:applications} presents the Case Study with
 \[
 V=\{x\,|\, x\geq 0,\ 1^\T x = 1,\, x\in \R^n\}\,.
 \] 
 Let us denote $\rho_q(x) = R_q(\xi^\T x)$ and $\phi_B(x) = \Omega_B(\xi^\T x)$. We are interested in the following portfolio optimization problems,
 \begin{alignat*}{4}
(P1,\rho_q,b)& \quad \quad \quad \max E[\xi^\T x]& \quad &\text{ s.t.}\quad \rho_q(x) \leq b,\ x\in V& \\
(P1,\phi_B,v)& \quad \quad \quad \max E[\xi^\T x]& \quad &\text{ s.t.} \quad \phi_B(x) \geq v,\ x\in V& \\
(P2,\rho_q,r)& \quad \quad \quad \min \rho_q(x)& \quad &\text{ s.t.} \quad E[\xi^\T x] \geq r,\ x\in V& \\
(P2,\phi_B,r)& \quad \quad \quad \max \phi_B(x)& \quad &\text{ s.t.} \quad E[\xi^\T x] \geq r,\ x\in V& 
\end{alignat*}

\section{Portfolio Optimization Problems with Expectile and Omega Functions} 
\label{sec:omega}

This section shows that the portfolio optimization problems with expectile and omega functions are equivalent and generate the same optimal portfolios. Expectile and omega are inverse functions in the sense of Proposition~\ref{prop:omegaexpectile}; see also $(6)$ in \shortcite{Bellini2018}.

In this section, we restrict our attention to $z > 1$ because in this case $R_{(1+z)^{-1}}$ is a coherent, strictly expectation bounded risk measure. When $E[(X-B)_+] > E[(X-B)_-] = 0$, we define $\Omega_B(X) = \oo$.
\begin{proposition}
	Suppose $1 < z < \oo$. Then
	\[
	\Omega_B(X) = z \iff R_{(1+z)^{-1}}(X) = -B \text{ and } X \neq B\;.
	\]
	\label{prop:omegaexpectile}
\end{proposition}
\begin{proof}
	Define $q = (1+z)^{-1}$.
	\begin{align*}
	\Omega_B(X) = \frac{E[(X-B)_+]}{E[(X-B)_-]} = z = \frac{1-q}{q}  &\iff \\
	qE[(X-B)_+] - (1-q)E[(X-B)_-] = 0, \; X\neq B &\iff \\
	e_q(X) = B, \; X\neq B	&\iff R_{(1+z)^{-1}}(X) = -B, \; X\neq B \qedhere
	\end{align*}
\end{proof}

Constraints for expectile and omega are equivalent in the following sense. 
\begin{proposition}
	Suppose $1 < z < \oo$, then 
	\[
\{X\, | \, \Omega_B(X) \geq z  \,\text{ or }\, X=B 
\}\; = \; \{X\, |\,R_{(1+z)^{-1}}(X) \leq - B \}   \;.
	\]
	\label{prop:omegaexpectile_sets}
\end{proposition}
\begin{proof}
Define $q = (1+z)^{-1}$. We begin by showing the first set includes in the second one. If $X = B$, then by translation invariance, $R_q(X) = -B$. Suppose $\Omega_B(X) \geq z$ and $X \neq B$. If $\Omega_B(X)$ is infinite, then $E[(X-B)_+] > E[(X-B)_-] = 0$. If $\Omega_B(X)$ is finite, then
\[
\frac{E[(X-B)_+]}{E[(X-B)_-]} = \Omega_B(X) \geq z = \frac{1-q}{q}.
\]
Note that whether $\Omega_B(X)$ is finite or infinite,
\[
qE[(X-B)_+]-(1-q)E[(X-B)_-] \geq 0\;.
\]
Since $qE[(X+R_q(X))_+] - (1-q)E[(X+R_q(X))_-] = 0$ by \eqref{focexpectile}, we have $R_q(X) \leq -B$ by Lemma \ref{monoexpectile}.

We now show that the second set includes in the first one. Suppose $R_q(X) \leq -B$. Equation \eqref{prop:omegaexpectile_sets:eq1} follows from the first order condition of expectile \eqref{focexpectile} and Lemma \ref{monoexpectile}.
\begin{align}
qE[(X+R_q(X))_+] - (1-q)E[(X+R_q(X))_-] &= 0 \Longrightarrow \nonumber\\
	qE[(X-B)_+]-(1-q)E[(X-B)_-] &\geq 0 \label{prop:omegaexpectile_sets:eq1}
\end{align}
If $X=B$, the statement is shown. If $X \neq B$, either $E[(X-B)_+]$ or $E[(X-B)_-]$ is nonzero. $E[(X-B)_+] = 0$ contradicts \eqref{prop:omegaexpectile_sets:eq1}, so $E[(X-B)_+] > E[(X-B)_-]$, which in conjunction with \eqref{prop:omegaexpectile_sets:eq1} implies $\Omega_B(X) \geq z$.
\end{proof}

Let $\mathcal{X}$ be a feasible set of random variables and $X\in \mathcal{X}$. Consider the following optimization problems in the space of random variables with omega and negative expectile constraints.

 \begin{alignat*}{4}
(P1,R_{(1+z)^{-1}},-B)& \quad \quad \quad \quad \max E[X]& \quad &\text{ s.t.}\quad R_{(1+z)^{-1}}(X) \leq - B,\; X\in \mathcal{X}& \\
(P1,\Omega_B,z)& \quad \quad \quad \quad \max E[X]& \quad &\text{ s.t.} \quad \Omega_B(X) \geq z,\; X\in \mathcal{X}& 
\end{alignat*}

The following corollary shows that problems  $(P1,R_{(1+z)^{-1}},-B)$ and  
$(P1,\Omega_B,z)$ with expectile and omega constraints are equivalent.
\begin{corollary}
	Suppose $z > 1$, and there exists $X_0 \in \mathcal{X}$ such that $\Omega_B(X_0) \geq z$. Then the problems $(P1,R_{(1+z)^{-1}},-B)$ and  
$(P1,\Omega_B,z)$ are equivalent, i.e. their optimal solution sets and objective values coincide.
\label{equival}
\end{corollary}

\begin{proof}
	$\Omega_B(X_0) \geq z > 1$ implies that $E[(X_0-B)_+] > E[(X_0-B)_-]$, which in turn implies $E[X_0] > B$. By Proposition~\ref{prop:omegaexpectile_sets}, $\Omega_B(X_0) \geq z$ implies $R_{(1+z)^{-1}}(X_0) \leq -B$. Hence, for both optimization problems in question, the optimal objective value is strictly larger than $B$ and replacing $\mathcal{X}$ with $\mathcal{X} \setminus B$ does not affect the optimal solution sets or objective values. The corollary now follows directly from Proposition~\ref{prop:omegaexpectile_sets}.
\end{proof}

Further we prove two propositions showing the relation of the following optimization problems $(P2,R_{(1+z)^{-1}},r)$ and $(P2,\Omega_B,r)$ with percentile and omega objectives.
 \begin{alignat*}{4}
(P2,R_{(1+z)^{-1}},r)& \quad \quad \quad \quad \min R_{(1+z)^{-1}}(X)& \quad &\text{ s.t.} \quad E[X] \geq r,\ X\in \mathcal{X}& \\
(P2,\Omega_B,r)& \quad \quad \quad \quad \max \Omega_B(X) & \quad &\text{ s.t.} \quad E[X] \geq r,\ X\in \mathcal{X}& 
\end{alignat*}

\begin{proposition}
	Let $z > 1$ and $X_0\in \mathcal{X}$ be an optimal solution of the optimization problem 	
	$(P2,R_{(1+z)^{-1}},r)$. If $X_0$ is non-constant, then $X_0$ is an optimal solution of the optimization problem $(P2,\Omega_{-R_{(1+z)^{-1}}(X_0)},r)$.
	\label{prop:exp=>omega}
\end{proposition}
\begin{proof}
Define $B = -R_{(1+z)^{-1}}(X_0)$. $X_0$ is non-constant and $R_{(1+z)^{-1}}(X_0) = -B$ so by Proposition~\ref{prop:omegaexpectile}, $\Omega_B(X_0) = z$. Suppose by contradiction that there exists $X_1 \in \mathcal{X}$ such that $E[X_1] \geq r$ and $\Omega_B(X_1) > \Omega_B(X_0)$.

If $\Omega_B(X_1)$ is infinite, then $E[(X_1-B)_+] > E[(X_1-B)_-] = 0$. This implies by the first order condition of expectiles \eqref{focexpectile} that $R_{(1+z)^{-1}}(X_1) < -B$, contradicting the optimality of $X_0$ as a solution of $(P2,R_{(1+z)^{-1}},r)$.

We now assume $\Omega_B(X_1) < \oo$ and introduce the following notation,
\[
z = \Omega_B(X_0),\ 
z' = \Omega_B(X_1),\ 
q = (1+z)^{-1},\ 
q' = (1+z')^{-1} \;.
\]

Note that $z' > z$ by assumption and consequently $q > q'$. By Proposition~\ref{prop:omegaexpectile}, we have $R_q(X_0) = R_{q'}(X_1) = -B$ and $X_1 \neq -B$.
The first order condition of level-$q'$ expectile implies
\[
q'E[(X_1 + R_{q'}(X_1))_+] - (1-q')E[(X_1 + R_{q'}(X_1))_-] = 0.
\]
Substituting $-B$ for $R_{q'}(X_1)$ gives 
\begin{equation}
q'E[(X_1 - B)_+] - (1-q')E[(X_1 - B)_-] = 0. \label{prop:exp=>omega:eq1}
\end{equation}
Equation \eqref{prop:exp=>omega:eq1} implies the following equation~\eqref{prop:exp=>omega:eq2},  since $q > q'$ and $X_1 \neq B$.
\begin{equation}
qE[(X_1 - B)_+] - (1-q)E[(X_1 - B)_-] > 0  \label{prop:exp=>omega:eq2}
\end{equation}
Applying the first order condition for level-$q$ expectile gives the next equation \eqref{prop:exp=>omega:eq3}.
\begin{align}
qE[(X_1 + R_q(X_1))_+] - (1-q)E[(X_1 + R_q(X_1))_-] = 0 \label{prop:exp=>omega:eq3}
\end{align}
Together, \eqref{prop:exp=>omega:eq2} and \eqref{prop:exp=>omega:eq3} imply that $-B > R_q(X_1)$, which contradicts the optimality of $X_0$ as a solution of $(P2,R_{(1+z)^{-1}},r)$.
\end{proof}

\begin{proposition}
	Let $X_0\in \mathcal{X}$ be an optimal solution of the optimization problem 	
	$(P2,\Omega_B,r)$. If $\Omega_B(X_0)$ is finite, then $X_0$ is an optimal solution of the optimization problem $(P2,R_{(1+\Omega_B(X_0))^{-1}},r)$.
	\label{prop:omega=>exp}
\end{proposition}
\begin{proof}
Define $q = (1+\Omega_B(X_0))^{-1}$ and suppose by contradiction there exists $X_1 \in \mathcal{X}$ such that $E[X_1] \geq r$ and $R_q(X_1) < R_q(X_0)$. The first order condition of level-$q$ expectile \eqref{focexpectile} gives \eqref{prop:omega=>exp:eq1}, which implies \eqref{prop:omega=>exp:eq1.5} since $R_q(X_1) < R_q(X_0)$.
\begin{align}
qE[(X_1 + R_q(X_1))_+] - (1-q)E[(X_1 + R_q(X_1))_-] &= 0\;, \label{prop:omega=>exp:eq1} \\
qE[(X_1 + R_q(X_0))_+] - (1-q)E[(X_1 + R_q(X_0))_-] &> 0\;. \label{prop:omega=>exp:eq1.5}
\end{align}
By Proposition~\ref{prop:omegaexpectile}, we have that $R_q(X_0) = -B$, which implies
\begin{equation}
qE[(X_1 - B)_+] - (1-q)E[(X_1 - B)_-] > 0\;. \label{prop:omega=>exp:eq2}
\end{equation}
Equation~\eqref{prop:omega=>exp:eq2} in turn implies the following equation~\eqref{prop:omega=>exp:eq3}, which is a contradiction to the optimality of $X_0$ as a solution of $(P2,\Omega_B,r)$.
\begin{equation}
\Omega_B(X_1) = \frac{E[(X_1 - B)_+]}{E[(X_1 - B)_-]} > \frac{1-q}{q} = \Omega_B(X_0) \;.\label{prop:omega=>exp:eq3}	\qedhere
\end{equation}
\end{proof}

Corollary \ref{equival} and Propositions
\ref{prop:exp=>omega}, \ref{prop:omega=>exp} show that
the same set of solutions (efficient frontier) can be generated by varying parameters in problems $(P1,R_{(1+z)^{-1}},-B)$,
$(P1,\Omega_B,z)$,
$(P2,R_{(1+z)^{-1}},r)$,
$(P2,\Omega_B,r)$.

Corollary \ref{equival} is valid for optimal portfolios found with  the problems  $(P1,\rho_{(1+z)^{-1}},-B) $ and $(P1,\phi_B,z)$.
\begin{corollary}
	If $z > 1$ and there exists $x_0 \in V$ such that $\phi_B(x_0) \geq z$, then problems $(P1,\rho_{(1+z)^{-1}},-B) $ and $(P1,\phi_B,z)$ are equivalent, i.e., their optimal solution sets and objective values coincide.
\end{corollary}

Similar corollaries follow from Propositions \ref{prop:exp=>omega} and \ref{prop:omega=>exp}. 
\begin{corollary}
\label{cor:exp=>omega}
	Let $z > 1$ and $x_0\in V$ be an optimal solution of the optimization problem 	
	$(P2,\rho_{(1+z)^{-1}},r)$. If $\xi^T x_0$ is non-constant, then $x_0$ is an optimal solution of the optimization problem $(P2,\phi_{-\rho_{(1+z)^{-1}}(x_0)},r)$.
\end{corollary}

\begin{corollary}
	Let $x_0\in V$ be an optimal solution of the optimization problem 	
	$(P2,\phi_B,r)$. If $\phi_B(x_0)$ is finite, then $x_0$ is an optimal solution of the optimization problem $(P2,\rho_{(1+\phi_B(x_0))^{-1}},r)$.
\end{corollary}

\section{Subgradient of $\rho_q(x)$}
\label{sec:subgrad}
The function $\rho_q(x) = R_q(\xi^\T x)$ is convex in $x$ for  $q \leq \frac{1}{2}$.
This section presents a formula for subgradients of the convex function $\rho_q(x)$ when $q < \frac{1}{2}$. The subgradient formula can be used for implementation of the convex programming algorithms for minimization of $\rho_q(x)$.

\begin{lemma}
	Suppose $q < \frac{1}{2}$. Then $R_q(X) > E[-X]$ for all non-constant $X\in L^1(\mathcal A)$, i.e. $R_q(X)$ is strictly expectation bounded as defined in \shortciteN{rockafellar2006generalized}.
\end{lemma}

\begin{proof}
Since $0 = E[X - E[X]] = E[(X-E[X])_+] - E[(X-E[X])_-]$, we have
\[
\frac{1}{2}E[(X-E[X])_+] - \frac{1}{2}E[(X-E[X])_-] = 0.
\]
Since $X$ is non-constant, $E[(X-E[X])_+]$ and $E[(X-E[X])_-]$ are positive. Moreover, the assumption that $q < \frac{1}{2}$ implies
\[
qE[(X-E[X])_+] - (1-q)E[(X-E[X])_-] < 0.
\]
The expression $qE[(X-m)_+)] - (1-q)E[(X-m)_-]$ is non-increasing in $m$ and equals $0$ when $m = e_q(X)$. Hence, $e_q(X) < E[X]$ which implies $R_q(X) = -e_q(X) > E[-X]$.
\end{proof}
\shortciteN{bellini2014generalized} provide the following dual representation of $R_q(X)$.
\begin{align*}
R_q(X) &= -\inf_{Q \in \mathcal{Q}}E[XQ],\\
\mathcal{Q} &= \left\{Q\in L^\infty : Q > 0 \ a.s.,\ E[Q] = 1,\  \frac{\text{ess sup}(Q)}{\text{ess inf}(Q)} \leq \frac{1-q}{q}\right\}
\end{align*}

Because $R_q$ is strictly expectation bounded, $D_q(X) = R_q(X) + E[X]$ is a deviation measure in the sense of \shortciteN{rockafellar2006generalized}. Applying Proposition 1 of \shortciteN{rockafellar2006optimality} to $D_q(X)$ leads to
\[
\partial R_q(X) = -\text{argmin}_{Q \in \mathcal{Q}} E[XQ].
\]

The following proposition is modified from \shortciteN{bellini2014generalized}. Let us denote
\[
\nu_t(X)=q1_{\{X > e_q(X)\}} + (1-q)1_{\{X < e_q(X)\}} + t1_{\{X = e_q(X)\}}\;.
\]

\begin{proposition}[\shortciteNP{bellini2014generalized}]
	\label{subgradexpectile}
	For all $t \in [q, 1-q]$, 
	\[
	-Q_t = -\frac{\nu_t( X) }{E[\nu_t( X)]}
	\]
	is an element of $\partial R_q(X)$.
\end{proposition}

\begin{proof}
	First note that $Q_t \in \mathcal{Q}$, so to show $Q_t \in \partial R_q(X)$ it suffices to show $E[XQ_t] = e_q(X)$.
	\begin{align*}
	E[XQ_t] &= \\
	\frac{E[qX1_{\{X > e_q(X)\}} + (1-q)X1_{\{X < e_q(X)\}} + tX1_{\{X = e_q(X)\}}]}{E[q1_{\{X > e_q(X)\}} + (1-q)1_{\{X < e_q(X)\}} + t1_{\{X = e_q(X)\}}]}  &=\\
	e_q(X) + \frac{E[q(X - e_q(X))1_{\{X > e_q(X)\}} + (1-q)(X - e_q(X))1_{\{X < e_q(X)\}}]}{E[q1_{\{X > e_q(X)\}} + (1-q)1_{\{X < e_q(X)\}} + t1_{\{X = e_q(X)\}}]} &= e_q(X)
	\end{align*}
	The numerator of the second term in the penultimate step is $0$ because of the first order condition of expectile \eqref{focexpectile}.
\end{proof}

In Section \ref{sec:applications}, we solve portfolio optimization problems with a constraint on $\rho_q(x)$. The following lemma furnishes a map from the subdifferential of $R_q(X)$ to that of $\rho_q(x)$. We note that these optimization problems could alternatively be solved with convex programming approaches using the subdifferential of $R_q(X)$ and the theory of \shortciteN{rockafellar2006optimality}.

\begin{lemma}
	\label{subdiffmap}
	If $Y \in \partial R_q(\xi^\intercal x)$, then $(E[\xi_1 Y], \dots, E[\xi_n Y]) \in \partial \rho_q(x)$.
\end{lemma}
\begin{proof}
\[
	\rho_q(z) - \rho_q(x)	= R_q(\xi^\intercal z) - R_q(\xi^\intercal x) 
	\geq E[(\xi^\intercal z - \xi^\intercal x)Y] 
	= \sum_{i=1}^n (z_i - x_i)E[\xi_i Y] 
	= \langle (E[\xi_1 Y], \dots, E[\xi_n Y]), z-x \rangle \qedhere
	\]
\end{proof}

The following proposition provides a collection of subgradients for $\rho_q(x)$.
\begin{proposition}
	\label{prop:gr_formula}
	The following is a subgradient of $\rho_q(x)$ at $x$ for every $t\in [q,1-q]$,
	$$
	g_t =-\frac{E[\xi \,\nu_t( \xi^\intercal x) ]}{E[\nu_t( \xi^\intercal x)]}\;.
	$$
\end{proposition}

\begin{proof}
	Apply Lemma \ref{subdiffmap} to $-Q_t$ from Proposition \ref{subgradexpectile}.
\end{proof}

\section{Derivative of $\rho_q(\xi^\T x)$ for Discrete Distribution}
\label{sec:derivative}

In this section and in the following Case Study in Section \ref{sec:applications} we assume that the random vector $\xi$ follows a discrete distribution. We derive subgradient formula for $\rho_q(x)$ in this special case.
These results are used for the implementation of optimization algorithms.

\subsection{Notation}

\label{sec:notation}
Let $\xi$ be a random vector with finitely many outcomes $\xi^1, \dots, \xi^J$ with corresponding probabilities $p_j > 0$. For any $x \in \R^n$ define $P_x$, $N_x$, and $Z_x$ as follows.
\[
P_x = \{j\ |\ \xi^{j\T} x + \rho_q(x) > 0\},\quad N_x = \{j\ |\ \xi^{j\T} x + \rho_q(x) < 0\}, \quad Z_x = \{j\ |\ \xi^{j\T} x + \rho_q(x) = 0\}
\]

The following proposition provides a collection of subgradients for $\rho_q(x)$.
\begin{proposition}
	\label{prop:gr_discrete}
	The following is a subgradient of $\rho_q(x)$ at $x$ for every $t\in [q,1-q]$.
	\[
	g_t = -\frac{q\sum_{j \in P_x} p_j\xi^{j} + (1-q)\sum_{j\in N_x} p_j\xi^{j} + t\sum_{j\in Z_x}p_j\xi^j}{q\sum_{j\in P_x}p_j + (1-q)\sum_{j\in N_x}p_j + t\sum_{j \in Z_x}p_j}\;.
	\]
\end{proposition}
\begin{proof}
	Apply Proposition \ref{prop:gr_formula}.
\end{proof}

\subsection{Partial Derivative of $\rho_q(x)$ when $Z_x=\emptyset$}
This section again derives the gradient formula for discrete distributions without using Proposition \ref{prop:gr_discrete}. This is done for illustrative purposes to show that the result can be obtained directly without using the sophisticated dual representation concept. In the following, we use the fact that $R_q(X)$ can be defined as the unique solution to $qE[(X+m)_+)] - (1-q)E[(X+m)_-] = 0$. 

\begin{proposition}
	Suppose $Z_x = \emptyset$. Then there exists an open neighborhood $U$ of $x$ such that for every $y \in U$, 
	\[
	\rho_q(y) = -\frac{q\sum_{j\in P_x} p_j \xi^{j\T}y + (1-q)\sum_{j \in N_x} p_j \xi^{j\T}y}{q(\sum_{j \in P_x} p_j) + (1-q)(\sum_{j \in N_x} p_j)}.
	\]
\end{proposition}

\begin{proof}
By assumption, $P_x \sqcup N_x = \{1,\dots, J\}$, and by the definition of $\rho_q(x)$, we have
\[
q\sum_{j\in P_x} p_j (\xi^{j\T}x + \rho_q(x)) + (1-q) \sum_{j\in N_x}  p_j (\xi^{j\T}x + \rho_q(x)) = 0.
\]
After rearrangement, we have
\[
\rho_q(x) = -\frac{q\sum_{j\in P_x} p_j \xi^{j\T}x + (1-q)\sum_{j \in N_x} p_j \xi^{j\T}x}{q(\sum_{j \in P_x} p_j) + (1-q)(\sum_{j \in N_x} p_j)}.
\]

Note that $\rho_q(x)$ is convex and thus continuous. Choose open neighborhoods $U$ of $x$ and $V$ of $\rho_q(x)$ such that

\begin{enumerate}
\item $\{j\ |\ \xi^{j\T} y + m > 0 \}  = P_x$ for every $(y,m) \in U \times V$, 
\item $\{j\ |\ \xi^{j\T} y + m < 0 \} = N_x$ for every $(y,m) \in U \times V$, 
\item $-\frac{q\sum_{j\in P_x} p_j \xi^{j\T}y + (1-q)\sum_{j \in N_x} p_j \xi^{j\T}y}{q(\sum_{j \in P_x} p_j) + (1-q)(\sum_{j \in N_x} p_j)} \in V$ for every $y \in U$. 
\end{enumerate}
Let $y$ be an element of $U$.
\begin{align*}
-\frac{q\sum_{j\in P_x} p_j \xi^{j\T}y + (1-q)\sum_{j \in N_x} p_j \xi^{j\T}y}{q(\sum_{j \in P_x} p_j) + (1-q)(\sum_{j \in N_x} p_j)} &= m \iff \\
q\sum_{j\in P_x} p_j (\xi^{j\T}y + m) + (1-q)\sum_{j \in N_x} p_j (\xi^{j\T}y + m) &= 0 \iff \\
 q\sum_{\{j\ |\ \xi^{j\T} y + m > 0 \}} p_j (\xi^{j\T}y + m) + (1-q)\sum_{\{j\ |\ \xi^{j\T} y + m < 0 \}} p_j (\xi^{j\T}y + m) &= 0 \iff \\
 qE[(\xi^{j\T}y+m)_+)] - (1-q)E[(\xi^{j\T}y+m)_-] &= 0
\end{align*}

Hence, $m = \rho_q(y)$.
\end{proof}

\begin{corollary}
Suppose $Z_x = \emptyset$. Then 
\[
\frac{\partial \rho_q}{\partial x_i}(x) = -\frac{q\sum_{j\in P_x} p_j \xi^j_i + (1-q)\sum_{j \in N_x} p_j \xi^j_i}{q(\sum_{j \in P_x} p_j) + (1-q)(\sum_{j \in N_x} p_j)} \;.
\]
\end{corollary}

Note that $\rho_q(x)$ need not be differentiable when $Z_x \neq \emptyset$. As an example, suppose $n = 2$ and $\xi$ is uniformly distributed over $\{ (0,0), (\frac{1}{6}, \frac{1}{3}), (1, 1) \}$. Let $q = \frac{1}{4}$ and $x = (1,1)$. One can easily check that $\rho_q(x) = -\frac{1}{2}$. However, $\lim_{h \to 0^+} (\rho(x+he_1) - \rho(x))/h = -\frac{3}{14}$ while $\lim_{h \to 0^-} (\rho(x+he_1) - \rho(x))/h = -\frac{7}{30}$, so $\frac{\partial \rho_q}{\partial x_1}$ doesn't exist at $x=(1,1)$.

\section{Case Study}
\label{sec:applications}

The codes, data, and solution results for the first case study are available at ``Portfolio Optimization with Expectiles" \cite{casestudyexpectile}. We solved the problem in a MATLAB environment using the Portfolio Safeguard (PSG) optimization software. PSG has both linear and convex programming algorithms.
Convex programming is preferable for the cases with a large number of scenarios.

First, we consider the following single-stage stochastic optimization problem.
 \[
 \max E[\xi^\T x]\quad \text{ s.t.}\quad \rho_q(x) \leq b,\ x\geq 0,\ 1^\T x = 1 \;.
 \]
Here $\xi = (\xi_1,\xi_2,\xi_3,\xi_4)$ is a discrete, uniform random vector with $10,000$ outcomes and $q=0.05$. An optimal solution for each $b$ in a given range was computed by supplying PSG with subroutines calculating the value and a subgradient of $\rho_q$ at any decision vector $x$. \\

\begin{figure}[h]
\centering
\begin{subfigure}{.5\textwidth}
  \centering
  \includegraphics[width=0.9\linewidth]{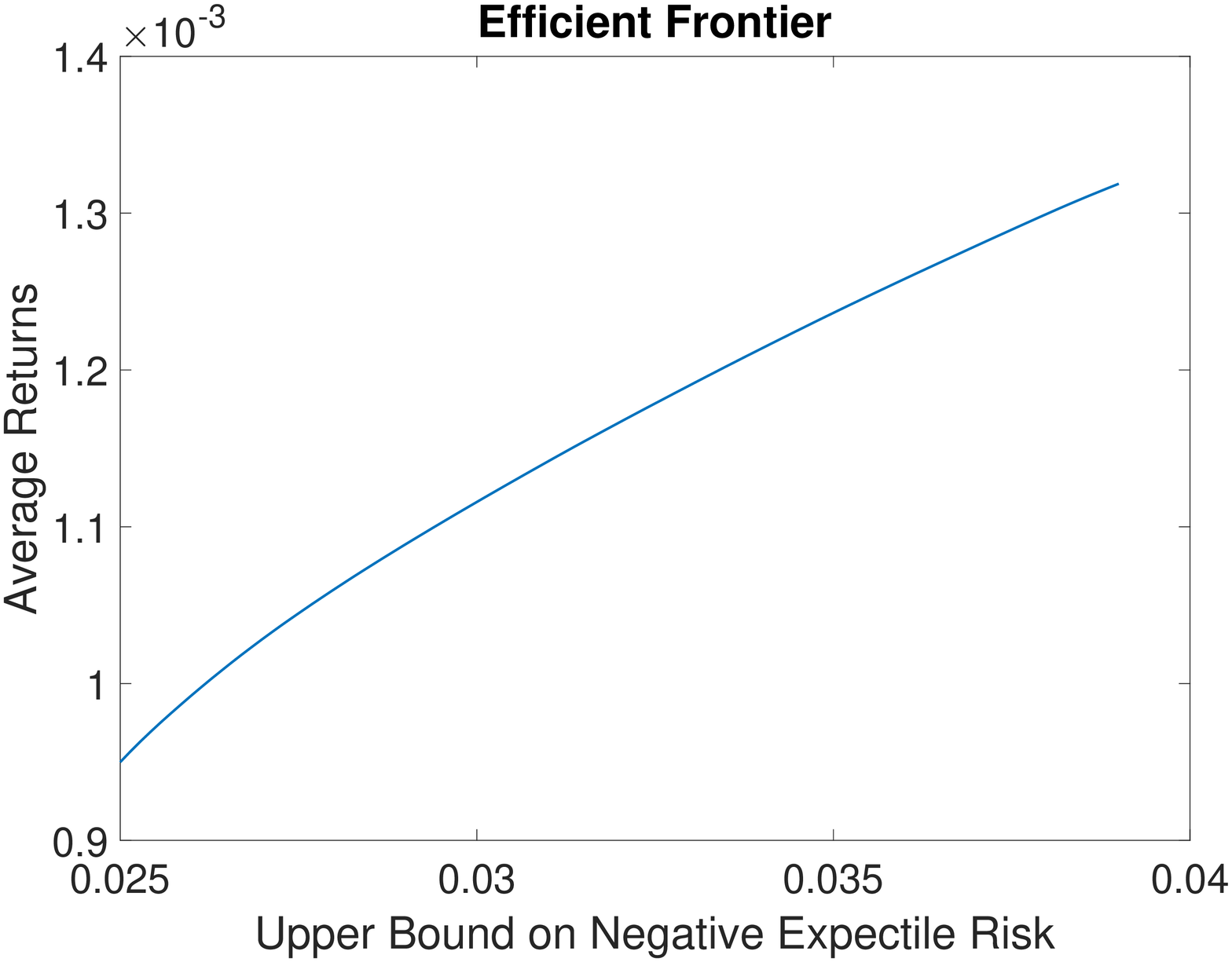}
\end{subfigure}%
\begin{subfigure}{.5\textwidth}
  \centering
  \includegraphics[width=0.9\linewidth]{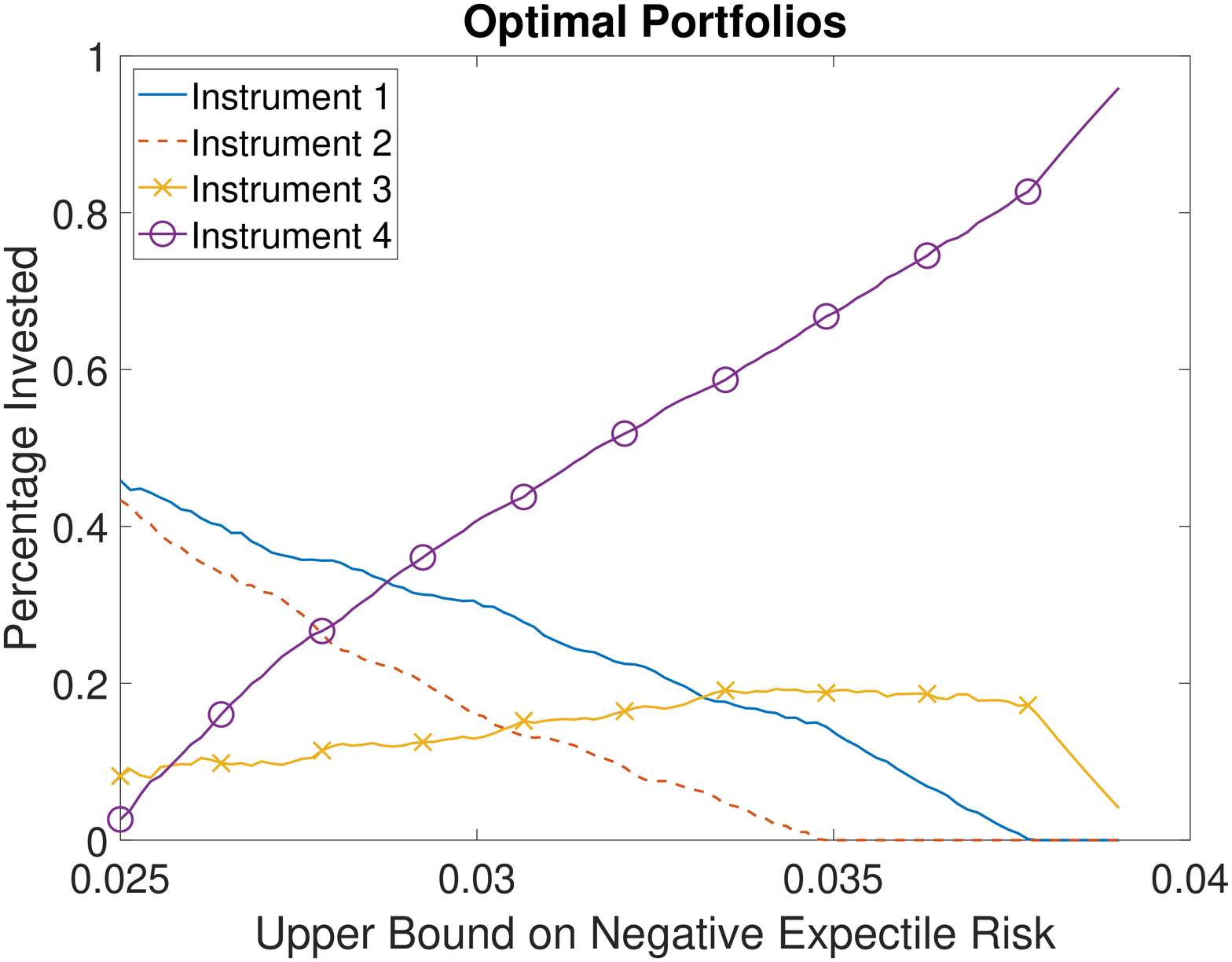}
\end{subfigure}
\caption{The expected returns of the optimal solution for a range of values of $b$, and the evolution of the optimal portfolio as the constraint on negative expectile risk is loosened.}
\label{fig:maxreturn}
\end{figure}

Figure \ref{fig:maxreturn} shows that when the upper bound on negative expectile risk is small, the optimal portfolio is a roughly even mixture of $\xi_1$ and $\xi_2$. However, as the constraint is loosened, the portfolio gradually shifts to being concentrated in $\xi_4$. This indicates that $\xi_4$ has the highest expected returns, while also carrying the most risk.

We also consider the problem of minimizing risk with expected returns bounded from below.
\[
 \min \rho_q(x) \quad \text{ s.t} \quad E[\xi^\T x] \geq r,\ x\geq 0,\ 1^\T x = 1 \;.
\]
 
 The efficient frontier in Figure \ref{fig:minrisk} shows that until we require an average return of more than $9 \times 10^{-4}$, the negative expectile risk of the optimal portfolio is constant. This suggests that in practice a lower bound greater than $9 \times 10^{-4}$ may be reasonable. Figure \ref{fig:minrisk} shows that when the constraint on average return is small, the optimal portfolio is an even mix of $\xi_1$ and $\xi_2$. However, as the constraint on average return increases, the portfolio shifts to the riskier $\xi_4$. \\
 
 \begin{figure}
\centering
\begin{subfigure}{.5\textwidth}
  \centering
  \includegraphics[width=.9\linewidth]{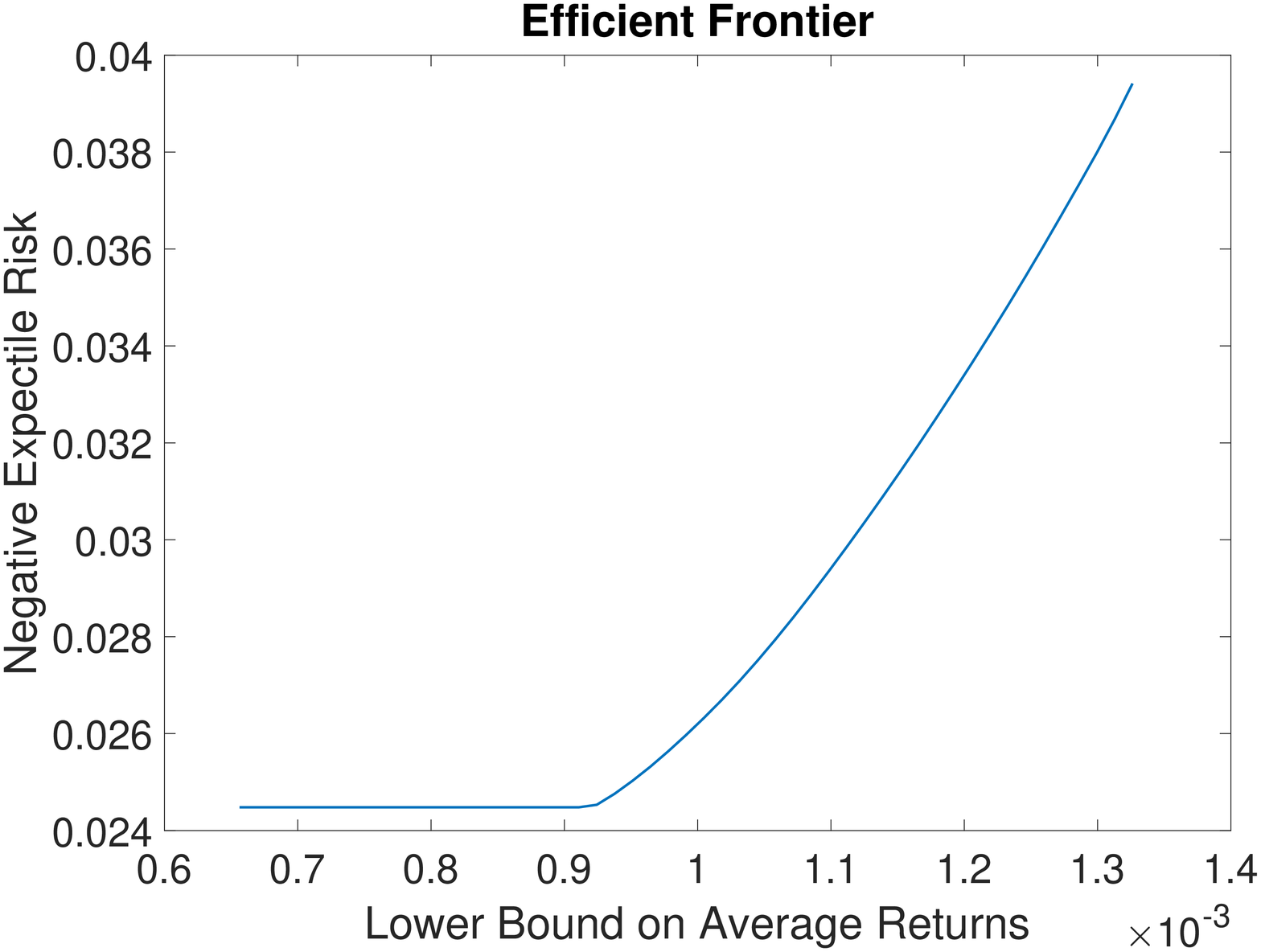}
\end{subfigure}%
\begin{subfigure}{.5\textwidth}
  \centering
  \includegraphics[width=.9\linewidth]{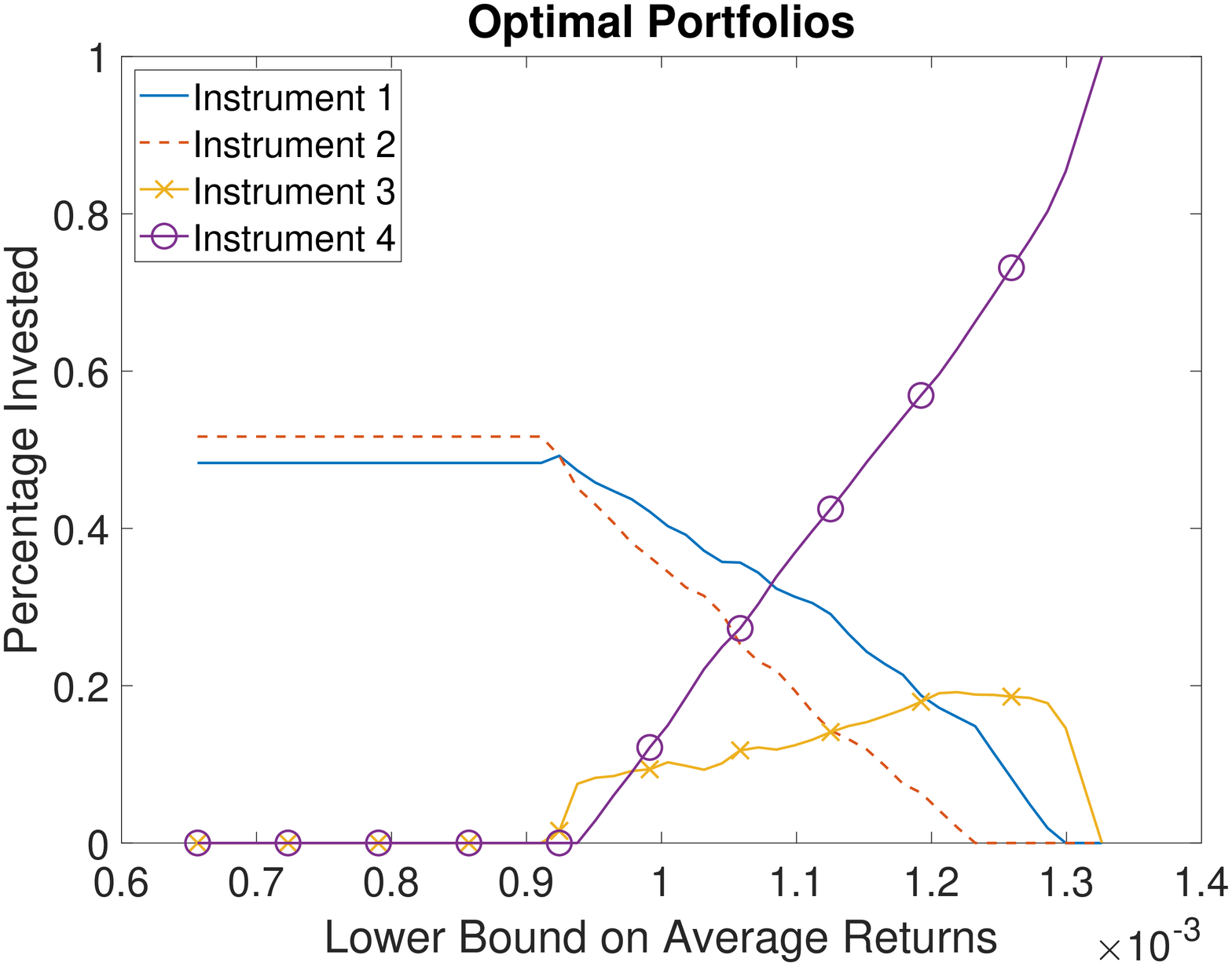}
\end{subfigure}
\caption{The negative expectile risk of the optimal solution for a range of values of $r$, and the evolution of the optimal portfolio as the constraint on average returns increases.}
\label{fig:minrisk}
\end{figure}

  Lastly, we compare our results with those from the  case study ``Basic CVaR Optimization Problem, Beyond Black-Litterman" \cite{casestudybasiccvar} which uses the same data.
 
\[
\min \text{CVaR}_{0.95}(-\xi^T x) \quad \text{ s.t.} \quad E[\xi^\T x] \geq 0.00105,\ x\geq 0,\ 1^\T x = 1 \;.
\]

The optimal solution is given by $x = (0.3669, 0.2574, 0.1304, 0.2454)$, which has an average return of $0.00105$ and $\text{CVaR}_{0.95}(-\xi^T x) = 0.054853$. The optimal solution with the same expected return constraint but minimizing negative expectile is given by $(x = 0.357588, 0.278139, 0.103928, 0.260346)$, which has an average return of $0.00105$ and $\rho_q(x) = 0.027668$.\\
 
 We also conducted experiments using data from the case study ``Omega Portfolio Rebalancing" \cite{casestudyomega}. The first experiment is a numerical test of Corollary \ref{cor:exp=>omega}. We set $z = 19$ such that $q = (1+z)^{-1} = 0.05$ and solved $(P2,\rho_{(1+z)^{-1}},r)$ for values of $r$ between $10^{-5}$ and $3 \times 10^{-4}$. For each $r$ and its corresponding optimal portfolio $x_0$, we then solved $(P2,\phi_{-\rho_{(1+z)^{-1}}(x_0)},r)$ and obtained an optimal portfolio $x_1$. According to Corollary \ref{cor:exp=>omega}, if $z > 1$, an optimal solution $x_0$ of $(P2,\rho_{(1+z)^{-1}},r)$ is also optimal for $(P2,\phi_{-\rho_{(1+z)^{-1}}(x_0)},r)$ as long as $\xi^T x_0$ is non-constant. Indeed, this result is supported by the experiment because for each $r$, the portfolios $x_0$ and $x_1$ were essentially equal. The largest difference in the percentage invested in any instrument between $x_0$ and $x_1$ was $0.017\%$ over all $r$ tested. We repeated the experiment with the dataset with $10,000$ scenarios from the previous case study and obtained similar results. We remark that the linear programming formulation of omega maximization due to \shortciteN{mausser2006optimising} took approximately $10$ seconds with this larger dataset while omega maximization with the convex programming approach afforded by PSG took approximately $0.5$ seconds.
 
 We also used this data to compare negative expectile to CVaR and VaR. For a range of values of $r$ between $10^{-5}$ and $3 \times 10^{-4}$, we solved the following optimization problems for $q = 0.05$ and $\alpha = 0.95$.
 
\[
 \min \rho_q(\xi^T x),\ CVaR_\alpha(-\xi^T x),\ \text{or }VaR_\alpha(-\xi^T x) \quad \text{ s.t} \quad E[\xi^\T x] \geq r,\ x\geq 0,\ 1^\T x = 1.
\]
 
 \begin{figure}[h]
\centering
\begin{subfigure}{.5\textwidth}
  \centering
  \includegraphics[width=.9\linewidth]{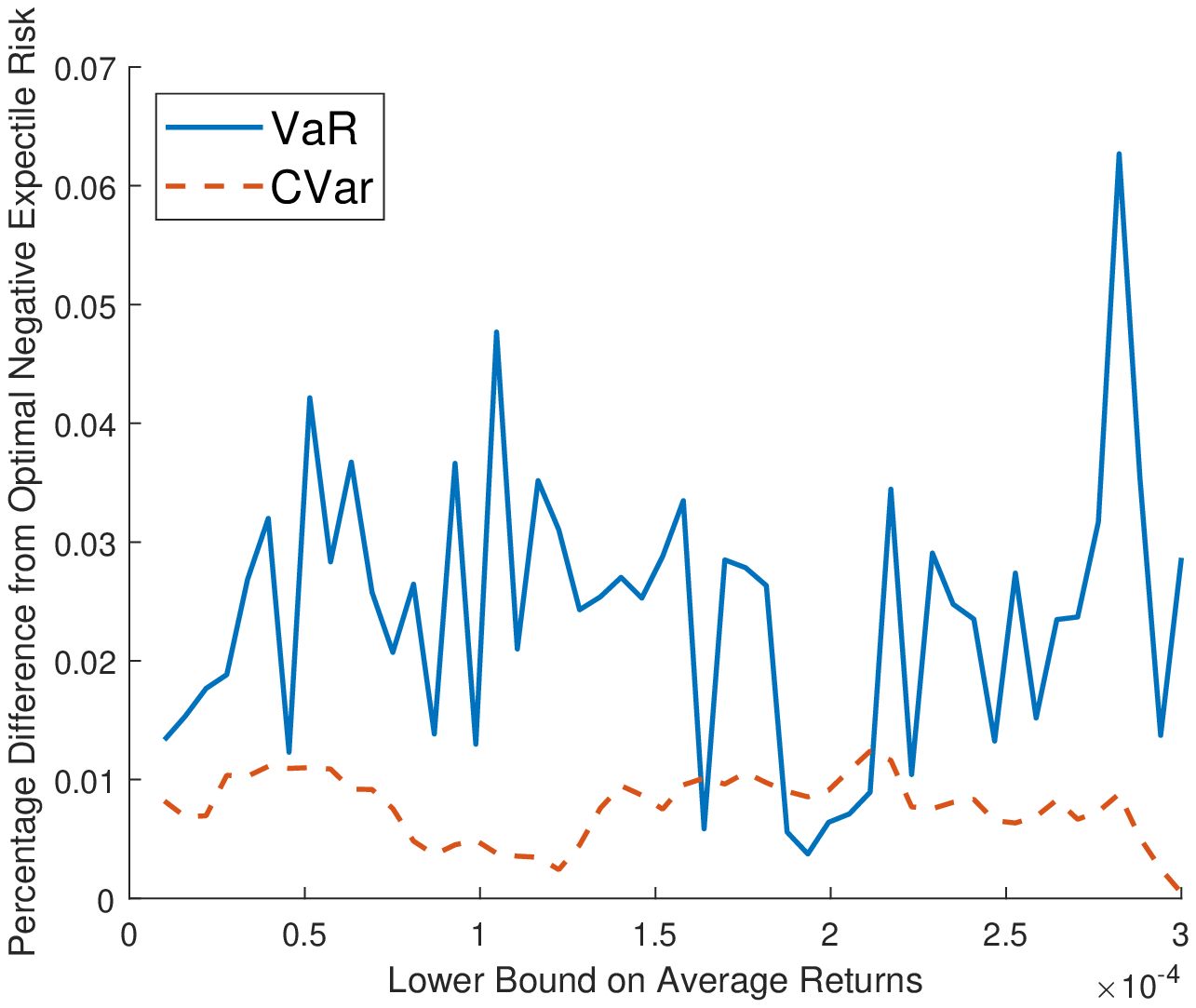}
\end{subfigure}%
\begin{subfigure}{.5\textwidth}
  \centering
  \includegraphics[width=.9\linewidth]{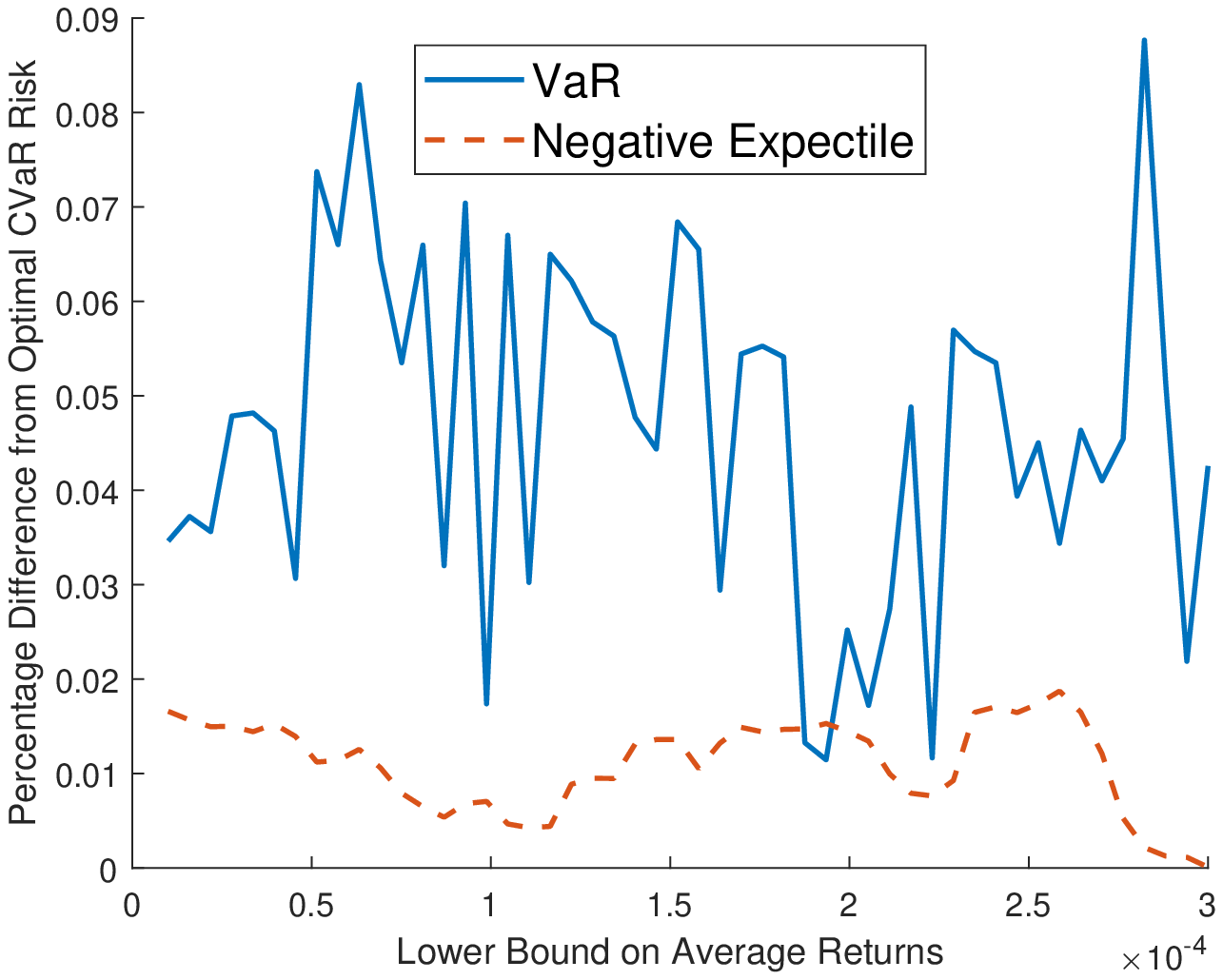}
\end{subfigure}

\begin{subfigure}{.5\textwidth}
  \centering
  \includegraphics[width=.9\linewidth]{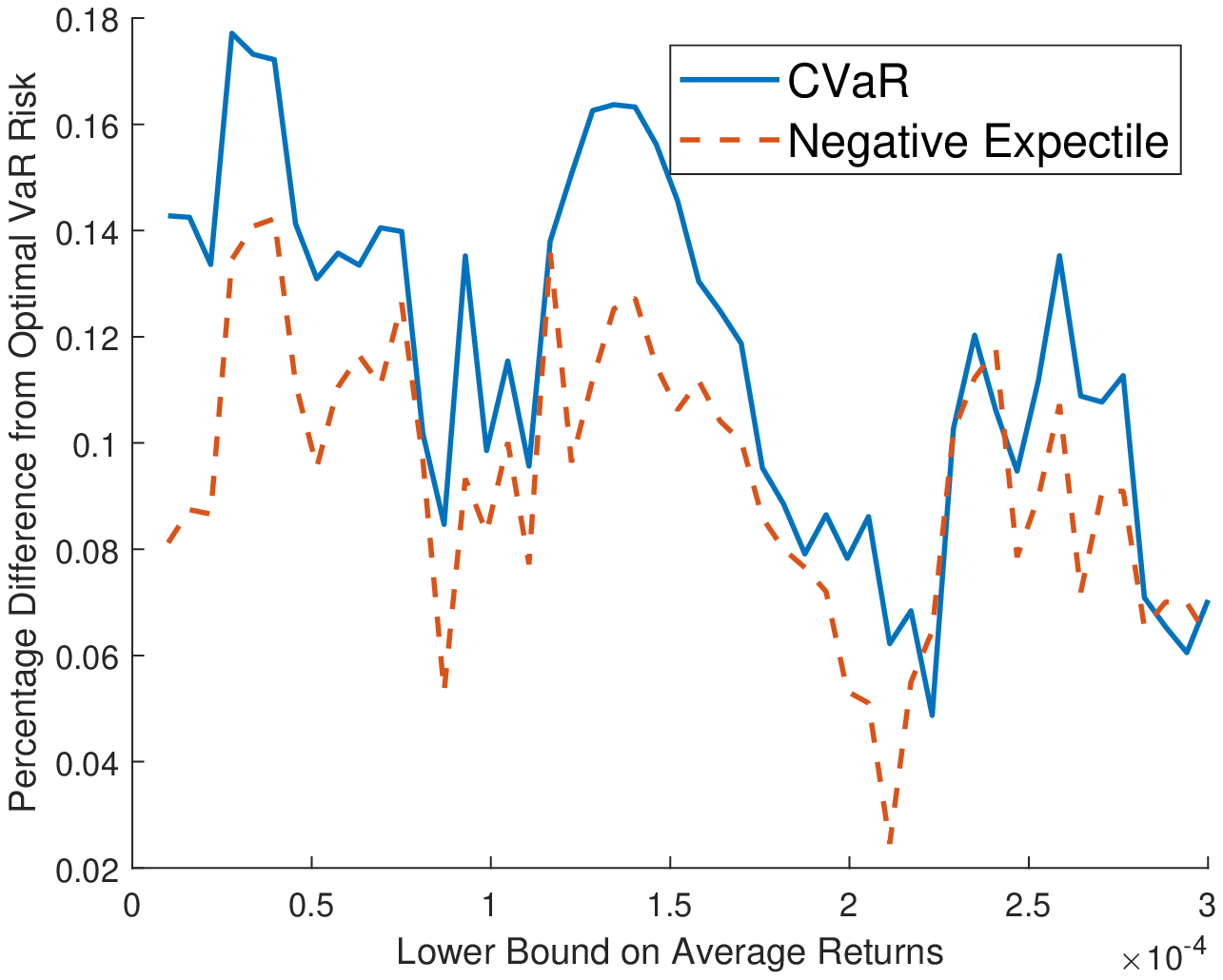}
\end{subfigure}
\caption{The relative difference between negative expectile, CVaR, and VaR of optimal portfolios.}
\label{fig:riskmeasurecomparison}
\end{figure}
 
 For each $r$, we obtained an optimal solution for each risk measure and computed the value of this solution with respect to the other two risk measures. Figure \ref{fig:riskmeasurecomparison} plots the relative differences. For example, the first plot shows the negative expectile risk of the optimal VaR and CVaR portfolios for a range of values of $r$. The plots show that the CVaR of the optimal expectile solution and the negative expectile of the optimal CVaR solution are both nearly optimal while the VaRs of the optimal CVaR and expectile solutions are considerably larger. Hence, this experiment suggests that CVaR and negative expectile may be more similar to each other than either is to VaR.
 
 \section{Conclusion}
 \label{sec:conclusion}
 
This paper considered the optimization problems of maximizing a portfolio's expected returns with an upper bound constraint on negative expectile risk or minimizing a portfolio's negative expectile risk with a lower bound constraint on expected returns. We proved equivalences between both of these problems and omega ratio optimization problems in Section \ref{sec:omega} by using the inverse relationship between expectile and the omega ratio described in Proposition \ref{prop:omegaexpectile}. In order to solve expectile portfolio optimization problems using convex programming in Section \ref{sec:applications}, we derived a subgradient for $\rho_q(x)$, the negative level-$q$ expectile risk of a portfolio $x$. This was done in Section \ref{sec:subgrad} by applying the theory of \shortciteN{rockafellar2006optimality} to the dual representation of negative expectile due to \shortciteN{bellini2014generalized} and again in Section \ref{sec:derivative} by an elementary argument. In Section \ref{sec:applications}, we also conducted a numerical test of an equivalence between expectile and omega ratio optimization and performed a comparison of negative expectile to two other popular risk measures, VaR and CVaR.

\end{document}